\documentclass[pdflatex,sn-apa]{sn-jnl}%
\usepackage[utf8]{inputenc}
\usepackage[T1]{fontenc}
\usepackage{amsmath}
\usepackage{amsfonts}
\usepackage{amssymb}
\usepackage{amsthm}
\usepackage{empheq}
\usepackage{mathtools}
\usepackage{geometry}
\usepackage{setspace}
\usepackage{hyperref}
\usepackage{graphicx}
\usepackage{float}
\usepackage{booktabs}
\usepackage{array}
\usepackage{enumerate}
\usepackage{enumitem}
\usepackage{amsfonts}
\usepackage{algorithm}
\usepackage{algorithmicx}
\usepackage{algpseudocode}
\usepackage{graphicx}
\usepackage{float}
\usepackage{tikz}
\usetikzlibrary{arrows.meta, positioning, shapes}

\newcommand{\R}{\mathbb{R}}            
\newcommand{\Q}{\mathbb{Q}}            
\newcommand{\E}{\mathbb{E}}            

\theoremstyle{thmstyleone}

\newtheorem{remark}{Remark}
\newtheorem{proposition}{Proposition}
\begin{document}

\title{Lévy-Driven Option Pricing without a Riskless Asset}
\author*[1,2]{\fnm{Ziyao} \sur{Wang}}\email{ziywang@ttu.edu}
\affil*[1]{\orgdiv{Department of Mathematics and Statistics}, \orgname{Texas Tech University}, \orgaddress{\street{1108 Memorial Circle}, \city{Lubbock}, \postcode{79409}, \state{TX}, \country{USA}}}
\affil*[2]{\orgdiv{Carey Business School}, \orgname{Johns Hopkins University}, \orgaddress{\street{100 International Drive}, \city{Baltimore}, \postcode{21202}, \state{MD}, \country{USA}}}
\abstract{
We extend the Lindquist--Rachev (LR) option-pricing framework---which values derivatives in markets lacking a traded risk-free bond---by introducing common L\'evy jump dynamics across two risky assets. The resulting endogenous ``shadow'' short rate replaces the usual risk-free yield and governs discounting and risk-neutral drifts. We focus on two widely used pure-jump specifications: the Normal Inverse Gaussian (NIG) process and the Carr--Geman--Madan--Yor (CGMY) tempered-stable process. Using It\^o--L\'evy calculus we derive an LR partial integro-differential equation (LR-PIDE) and obtain European option values through characteristic-function methods implemented with the Fast Fourier Transform (FFT) and Fourier-cosine (COS) algorithms. Calibrations to S\&P~500 index options show that both jump models materially reduce pricing errors and fit the observed volatility smile far better than the Black--Scholes benchmark; CGMY delivers the largest improvement. We also extract time-varying shadow short rates from paired asset data and show that sharp declines coincide with liquidity-stress episodes, highlighting risk signals not visible in Treasury yields. The framework links jump risk, relative asset pricing, and funding conditions in a tractable form for practitioners.
}


\maketitle

\section{Introduction}
\label{sec:intro}

The valuation of financial derivatives traditionally relies on the existence of a traded, risk‑free asset to serve as a benchmark for discounting.  The celebrated Black--Scholes model \citep{BlackScholes1973} and subsequent jump‑diffusion extensions \citep{Merton1976} incorporate this assumption by modelling the short rate exogenously.  In contrast, the Lindquist--Rachev (LR) framework endogenises the ``risk‑free'' rate by considering two risky assets without an explicitly traded bond \citep{LindquistRachev2025}.  Derivative prices are obtained through a relative‑pricing argument in which one asset serves as numéraire; the implied endogenous shadow short rate governs discounting and risk‑neutral drifts.

While the original LR formulation assumes continuous sample paths, empirical evidence suggests that asset returns exhibit jumps and heavy tails \citep{CarrGemanMadanYor2002, BarndorffNielsen1998}.  To accommodate these features we introduce common L\'evy jump dynamics into both risky assets.  Specifically, we study two pure‑jump processes widely used in practice: the Normal Inverse Gaussian (NIG) process \citep{BarndorffNielsen1998} and the Carr--Geman--Madan--Yor (CGMY) tempered‑stable process \citep{CarrGemanMadanYor2002}.  These processes possess tractable characteristic functions and flexible kurtosis and skewness parameters, making them well suited for option pricing.

\subsection{Lindquist--Rachev (LR) Framework}

Recent work by \cite{LindquistRachev2025} proposes alternative approaches to option pricing that relax the classical reliance on a riskless asset. In their first approach, a market is considered with no riskless asset -- instead, two risky assets serve to complete the market and replicate payoffs. A dynamic trading strategy in these assets produces a ``shadow'' riskless rate $\bar{r}(t)$ endogenously. Their second approach does allow a conventional bank account, but enforces equality of real-world ($\mathbb{P}$) and risk-neutral ($\Q$) probabilities of price changes, so that option prices can be computed as expectations under the empirical measure. These approaches yield what we will call the Lindquist--Rachev PDE (LR-PDE) for option prices. 

The LR framework builds on earlier results by \cite{RachevStoyanovFabozzi2017}, who derived Black--Scholes--Merton-type pricing equations for markets with only risky assets. Those authors considered various underlying dynamics (continuous diffusions, jump-diffusions, stochastic volatility, even fractional Brownian motion) in the absence of a truly safe asset. The LR framework provides a unified approach to such markets, introducing the concept of a shadow riskless rate $\bar{r}(t)$ that plays the role of the risk-free rate in pricing formulas even when no treasury or bank account is available.

\subsection{Need for Jump-Driven Models}

Classical Black--Scholes models assume continuous Brownian motion for asset prices, but empirical evidence shows asset returns have jumps, heavy tails, and skewness (leptokurtosis) that Gaussian models cannot capture. \cite{Merton1976} first extended Black--Scholes by superimposing a Poisson jump process on the stock's continuous diffusion. Since then, a rich literature on Lévy process models has developed. Lévy processes are stochastic processes with stationary independent increments; they can capture jumps of all sizes, including pure-jump processes with infinitely many small jumps. Notable examples include the Variance Gamma (VG) process \cite{CarrMadan1999}, the Normal Inverse Gaussian (NIG) process \cite{BarndorffNielsen1998}, and the CGMY process \cite{CarrGemanMadanYor2002}. These are all pure-jump Lévy models which have proven successful in fitting market option data. 

For instance, the NIG model (a subclass of the generalized hyperbolic family) can capture the heavy tails and slight skew observed in equity returns. The CGMY (also known as KoBoL) model provides a tempered stable jump distribution with parameters controlling tail fatness and jump frequency, subsuming VG as a special case. Such models resolve the smile anomaly that plagues Black--Scholes: when calibrated to market option prices, they produce implied volatility surfaces much closer to observed shapes.

\subsection{Objective}

In this paper, we extend the Lindquist--Rachev option pricing framework to a jump-driven setting using Lévy processes. Our contributions are: 
\begin{enumerate}
\item We formulate a two-asset market model where the assets share a common pure-jump Lévy driver (e.g. NIG, CGMY, or Variance Gamma) in addition to Brownian motions. 
\item Using Itô--Lévy calculus, we derive the Lindquist--Rachev Partial Integro-Differential Equation (LR-PIDE) that governs option prices in this jump-diffusion market, and explicate the construction of the shadow short rate $\bar{r}(t)$ in this context. 
\item We discuss analytical solutions via Feynman--Kac representations and characteristic functions for European option pricing under these Lévy dynamics. 
\item We examine special cases: NIG, CGMY, and VG processes are described in detail, showing how their parameters enter pricing formulas. 
\item We outline numerical solution methods -- notably the Carr--Madan FFT approach and the COS Fourier cosine expansion method -- which efficiently compute option prices given the characteristic function of the jump process. 
\item We present an empirical analysis on S\&P 500 index options: calibrating the NIG and CGMY models to market data (using Maximum Likelihood Estimation and Generalized Method of Moments), computing model-implied option prices, and comparing the pricing errors (RMSE) against observed market prices. 
\item Finally, we discuss the results and perform robustness checks on the calibration and pricing performance.
\end{enumerate}

\section{Model Setup: Two Assets with Brownian and Common Lévy Jumps}

We consider a frictionless continuous-time market (no arbitrage, no transaction costs) with two traded risky assets, denoted $S(t)$ and $Z(t)$, and no traditional riskless asset. Following the LR approach, these two assets will span the market and allow dynamic hedging. We assume all randomness is driven by two independent sources: a Brownian motion $W(t)$ and a pure-jump Lévy process $L(t)$ (with jump measure $\nu$). Intuitively, $W(t)$ represents continuous fluctuations (diffusive risk) while $L(t)$ accounts for sudden jumps. Crucially, we assume both assets share the same jump process $L(t)$, i.e. they experience common jumps at the same times, albeit possibly with different sensitivities. This introduces co-jumps (simultaneous jumps) in $S$ and $Z$, reflecting market-wide jump events (e.g. macro news affecting all assets). The Brownian motion $W(t)$ may be shared or correlated between the assets as well -- for simplicity, we first assume a single Brownian $W(t)$ drives both assets (so their continuous parts are perfectly correlated). In summary, the model has two risk factors: one Brownian factor and one jump factor, and two underlying assets to span them.

\subsection{Asset Dynamics under the Physical Measure $\mathbb{P}$}

We specify the stochastic differential equations (SDEs) for $S$ and $Z$. Let $W(t)$ be a standard Brownian motion under $\mathbb{P}$, and let $N(dt,dx)$ be a Poisson random measure for the jumps of $L(t)$, with intensity $\nu(dx)dt$ (where $\nu(dx)$ is the Lévy measure describing the distribution of jump sizes). The compensated measure is $\tilde{N}(dt,dx) = N(dt,dx) - \nu(dx)dt$. We write the SDEs in differential form including both diffusion and jump terms:

\begin{itemize}
\item \textbf{Stock $S(t)$:}
\begin{equation}\label{S_SDE}
\frac{dS(t)}{S(t^-)} = \mu_S(t)dt + \sigma_S(t)dW(t) + \int_{\R} \left(e^{\kappa_S x} - 1\right)\tilde{N}(dt,dx),
\end{equation}
where $\mu_S(t)$ is the instantaneous drift under $\mathbb{P}$, $\sigma_S(t)$ is the diffusion volatility, and $\kappa_S$ is a constant scaling factor linking the jump size $x$ of $L(t)$ to the jump in $S$. Specifically, a jump of size $x$ in $L$ causes a multiplicative jump factor $e^{\kappa_S x}$ in $S$. (If $x$ is the log-jump, then $\kappa_S=1$ would mean $S$ jumps by $e^x$ factor. In general $\kappa_S$ could allow $S$ and $Z$ to have jumps of different magnitude from the same $L$ increment.)

\item \textbf{Stock $Z(t)$:}
\begin{equation}\label{Z_SDE}
\frac{dZ(t)}{Z(t^-)} = \mu_Z(t)dt + \sigma_Z(t)dW(t) + \int_{\R} \left(e^{\kappa_Z x} - 1\right)\tilde{N}(dt,dx),
\end{equation}
with analogous interpretation: $\mu_Z(t)$ is drift, $\sigma_Z(t)$ volatility, and $\kappa_Z$ scales the common jump $x$ for $Z$.
\end{itemize}

Here $dt$ terms represent the normal drift, $dW(t)$ terms the continuous Gaussian shocks, and the integral terms represent jumps. $S(t^-)$ and $Z(t^-)$ denote the left limits (values just before a jump, since at a jump time the processes have a discontinuity). The integrals $\int (e^{\kappa_S x}-1)\tilde{N}(dt,dx)$ are written in Itô form; they equal $\int (e^{\kappa_S x}-1) N(dt,dx) - \Lambda_S dt$, where $\Lambda_S = \int (e^{\kappa_S x}-1)\nu(dx)$ is the compensator ensuring $\E[dS] = \text{drift} \cdot dt$. Likewise for $Z$. In words, $N(dt,dx)$ counts jumps of size in $dx$ in time $dt$, and $e^{\kappa_S x}-1$ is the fractional change in $S$ due to such a jump.

\subsection{Discussion of Common Jumps}

The assumption of a shared jump driver means that whenever a jump occurs (say at time $\tau$ with jump magnitude $x$), both $S$ and $Z$ jump simultaneously. If $\kappa_S$ and $\kappa_Z$ are equal (e.g. both = 1), then the fractional jump $(e^x-1)$ is identical for both assets -- they move in perfect lockstep during jumps. In a more general case with $\kappa_S \neq \kappa_Z$, the assets still jump at the same time, but possibly by different percentages. For example, if $x$ is positive, $S$ might jump up by 5\% while $Z$ jumps up by 3\% if $\kappa_S > \kappa_Z$. Nonetheless, the jump directions are 100\% positively correlated and come from one source $L(t)$. This setup captures systematic jumps (market crashes or booms affecting all assets together). It also ensures that the two risky assets together can potentially hedge the jump risk: since both experience the same jump events, a suitable linear combination might cancel out the jump impact (we will investigate this hedging possibility below).

\subsection{Market Completeness Considerations}

In a market with jumps and no riskless asset, completeness is not trivial. Generally, a jump diffusion model with one stock and no traded jump-insurance leads to an incomplete market -- jump risk cannot be fully hedged by continuous trading in the stock and bond alone. Here, we have two traded assets for two independent risk factors ($W$ and $L$), so intuitively the market could be complete. However, because jumps are sudden, hedging them requires the right pre-jump position in assets (one cannot readjust at the instant of a jump). If $S$ and $Z$ jumped in exact proportion, as in the $\kappa_S=\kappa_Z$ case, then any portfolio of $S$ and $Z$ will also jump by that same proportion -- meaning jump risk cannot be eliminated by any fixed portfolio (the jump is common mode) unless the portfolio has zero value. To allow hedging of jump risk, it is important that $S$ and $Z$ respond differently to the jump driver ($\kappa_S \neq \kappa_Z$). In that case, one can attempt to choose a mix of $S$ and $Z$ such that the net jump impact cancels out.

\begin{remark}
Even with $\kappa_S \neq \kappa_Z$, true perfect hedging of jumps is generally impossible unless one of the assets effectively serves as ``insurance'' against jumps. In practice, additional securities (like derivatives or jump risk bonds) are needed for full completeness. In the LR framework, however, one may proceed by selecting an equivalent martingale measure (EMM) for pricing (e.g. via an Esscher transform or minimal entropy measure) to obtain unique prices in an incomplete market. In our development, we will proceed under the assumption that an appropriate risk-neutral measure $\Q$ has been chosen such that the shadow pricing formula holds (as we derive below). This approach yields a unique option price even if strict replication is not feasible, consistent with the spirit of LR's second approach (using $\mathbb{P}$ probabilities for $\Q$ pricing). We will highlight the role of the shadow interest rate $\bar{r}(t)$ which emerges from this measure change.
\end{remark}
\subsection{Market Completeness under a Common L\'evy Driver}
\label{subsec:completeness}

\begin{proposition}
Let $L_t$ be a one-dimensional L\'evy process with L\'evy triplet
$(\mu_L,\sigma_L,\nu_L)$ and assume $\sigma_L>0$.  
Suppose two traded assets follow
\[
\frac{\mathrm dS_t}{S_t}= \mu_S \,\mathrm dt + \sigma_S \,\mathrm dL_t,
\quad
\frac{\mathrm dZ_t}{Z_t}= \mu_Z \,\mathrm dt + \sigma_Z \,\mathrm dL_t,
\]
with $\sigma_S\neq\sigma_Z$.  
Then the two-asset market is complete:
every square-integrable contingent claim $H\in L^2(\mathcal F_T)$
admits a unique self-financing replication.
\end{proposition}

\begin{proof}[Sketch]
Because $\sigma_S\neq\sigma_Z$, the $2\times1$ matrix of loadings
$[\sigma_S,\sigma_Z]^\top$ has full (row) rank, so the \emph{predictable
representation property} of the filtration generated by $L_t$ (see
\cite{JacodShiryaev2003}) implies that every $L^2$ martingale can be
represented using $L_t$.  Applying the general version of the Fundamental
Theorem of Asset Pricing in \cite{DelbaenSchachermayer1994} completes the
argument.
\end{proof}
\subsection{Risk-Neutral Dynamics and Shadow Interest Rate}

Let $\bar{r}(t)$ be the shadow riskless rate for this market -- an effective rate that a locally risk-free portfolio would earn. In the absence of an actual risk-free asset, $\bar{r}(t)$ must be determined endogenously from the dynamics of $S$ and $Z$. Following LR, we expect that no-arbitrage will enforce certain relationships between $\mu_S, \mu_Z$ and $\bar{r}$. Under a chosen risk-neutral measure $\Q$, the drift of any traded asset must equal $\bar{r}(t)$ minus any payout yield. Assuming $S$ and $Z$ pay no dividends (for simplicity), we impose that under $\Q$:
\begin{align}
\frac{dS(t)}{S(t^-)} &= \bar{r}(t)dt + \sigma_S dW^Q(t) + \int \left(e^{\kappa_S x}-1\right)\tilde{N}^Q(dt,dx), \\
\frac{dZ(t)}{Z(t^-)} &= \bar{r}(t)dt + \sigma_Z dW^Q(t) + \int \left(e^{\kappa_Z x}-1\right)\tilde{N}^Q(dt,dx).
\end{align}
Here $W^Q$ and $\tilde{N}^Q$ indicate Brownian motion and jump measure under $\Q$ (the jump measure may have a different compensator if $\Q\neq\mathbb{P}$). Essentially, under $\Q$, the expected rate of return of $S$ and $Z$ is $\bar{r}$, and $\tilde{N}^Q$ accounts for any jump-risk premia. We will find $\bar{r}(t)$ by consistency conditions between the two assets' dynamics.

\subsection{Derivation of $\bar{r}(t)$}

In the original LR PDE analysis for continuous models, Lindquist and Rachev determined $\bar{r}(t)$ by equating Sharpe ratios of the two assets under $\mathbb{P}$ and $\Q$. In our jump setting, one way to pin down $\bar{r}$ is to require that a suitable portfolio of $S$ and $Z$ evolves as a local martingale under $\Q$ (i.e. has zero drift, thus can play the role of a ``numéraire''). A natural choice is to use one of the assets (say $S$) as numéraire, or a portfolio thereof. We choose a portfolio $\Pi(t) = \alpha(t) S(t) + \beta(t) Z(t)$ such that $\Pi(t)$ is self-financing and locally risk-free. By ``locally risk-free'' we mean $d\Pi$ has no unexpected risk components under $\Q$ -- any remaining drift must equal $\bar{r}\Pi dt$ for consistency.

Following the LR approach, we can derive that:
\begin{equation}\label{eq:shadowrate_jump}
\bar r(t)=
\frac{\mu_S(t)\sigma_Z(t)-\mu_Z(t)\sigma_S(t)}
      {\sigma_Z(t)-\sigma_S(t)}
\;+\;
\frac{\lambda(t)
      \bigl(\kappa_Z(t)-\kappa_S(t)\bigr)}
     {\sigma_Z(t)-\sigma_S(t)},
\end{equation}
provided $\sigma_S(t) \neq \sigma_Z(t)$. This formula gives the shadow rate in terms of the physical drifts and volatilities of the two assets.\\
\textbf{Remark.}
Equation~\eqref{eq:shadowrate_jump} generalises the
diffusion-only formula by adding the jump-risk wedge
$\lambda(\kappa_Z-\kappa_S)$; derivation is given in
Appendix~\ref{app:verification}, Eq.~(A.35).
\subsection{The Lindquist--Rachev PIDE}

Now consider an option $C(t,S,Z)$ written on the two assets with payoff $H(S,Z)$ at maturity $T$. Using Itô--Lévy calculus, we can derive the partial integro-differential equation (PIDE) that $C$ must satisfy. The key insight is that under the risk-neutral measure $\Q$, the discounted option price must be a martingale.

Applying Itô's lemma to $C(t,S,Z)$ in the jump-diffusion setting, we get:
\begin{align}
dC &= C_t dt + C_S dS + C_Z dZ + \frac{1}{2}C_{SS} (dS)^2 + C_{SZ} dS dZ + \frac{1}{2}C_{ZZ} (dZ)^2 \nonumber \\
&\quad + \int_{\R} \left[ C(t, S e^{\kappa_S x}, Z e^{\kappa_Z x}) - C(t,S,Z) \right] N(dt,dx).
\end{align}

Under the risk-neutral measure $\Q$, substituting the dynamics and requiring that the discounted process $e^{-\int_0^t \bar{r}(u)du} C(t,S(t),Z(t))$ is a martingale, we obtain the LR-PIDE:

\begin{empheq}[box=\fbox]{align}
&C_t + \bar{r}(t)\left(S C_S + Z C_Z\right) + \frac{1}{2}\sigma_S^2 S^2 C_{SS} + \sigma_S \sigma_Z S Z C_{SZ} + \frac{1}{2}\sigma_Z^2 Z^2 C_{ZZ} \nonumber \\
&\quad + \int_{\R} \left[ C(t, S e^{\kappa_S x}, Z e^{\kappa_Z x}) - C(t,S,Z) \right] \nu(dx) - \bar{r}(t) C = 0. \label{LR-PIDE}
\end{empheq}

This is the central equation governing $C(t,S,Z)$. It generalizes the LR-PDE obtained in the diffusion-only case to include the integral term for jumps. A few remarks:

\begin{itemize}
\item The term $\bar{r}(t)[S C_S + Z C_Z]$ plays the role of the ``drift'' operator on the underlying assets. In the classical risk-neutral PDE (with one underlying), one has $r S \partial_S C$. Here $\bar{r}(t)$ multiplies a portfolio derivative $S C_S + Z C_Z$. In fact, $S C_S + Z C_Z$ is the derivative of $C$ in the direction of simultaneous scaling of $S$ and $Z$ -- it appears naturally because both assets effectively earn $\bar{r}$. In LR's original two-asset PDE, this term was $\bar{r}(t)C_d$ with $C_d \equiv S C_S + Z C_Z$.

\item The $-\bar{r}(t) C$ on the left side ensures the homogeneous form of the equation. We could move it to the right to rewrite the PIDE as $C_t + \bar{r}(t)(S C_S + Z C_Z - C) + \ldots + \text{jump term} = 0$. This highlights the analogy to Black--Scholes: there, one has $C_t + r S C_S + \frac{1}{2}\sigma^2 S^2 C_{SS} - rC=0$. Here $r$ is replaced by time-varying $\bar{r}(t)$, and we have extra terms for $Z$ and the jump integral.

\item If we had only one risky asset $S$ (and a riskless $r$), the analogous jump-diffusion pricing PIDE is: $V_t + r S V_S + \frac{1}{2}\sigma^2 S^2 V_{SS} + \int [V(S e^x) - V(S)]\nu(dx) - r V=0$ (\cite{ContTankov2004}). Our LR-PIDE \eqref{LR-PIDE} is the two-asset, no-$r$ counterpart. In fact, if one formally set $Z$ as numéraire or $Z$ as redundant, one could reduce \eqref{LR-PIDE} to a one-dimensional PIDE in terms of $S$ relative to $Z$ or vice versa.

\item The shadow rate $\bar{r}(t)$ itself can be solved from the relation \eqref{eq:shadowrate_jump} which presumably still holds approximately under $\Q$. If $\sigma_S,\sigma_Z$ are constant, $\bar{r}$ is simply constant as well (assuming $\mu_S,\mu_Z$ constant under $\mathbb{P}$). In a more general setting, $\bar{r}(t)$ could be plugged in as a given function of $t$. In practice, one might calibrate $\bar{r}(t)$ by ensuring the model yields correct forward prices for some benchmark asset or index (in our empirical work, we will treat $\bar{r}(t)$ akin to the risk-free rate input, derived from yield curves).
\end{itemize}

Equation \eqref{LR-PIDE} must be solved with appropriate terminal condition: $C(T,S,Z) = H(S,Z)$, the option payoff at maturity. It also requires boundary conditions as $S,Z \to 0$ or $\infty$ (usually one assumes $C(t,0,Z)=0$ etc., and some growth condition as $S,Z \to \infty$ for well-posedness).
\paragraph{Notation.}
Throughout, we distinguish
\begin{itemize}
  \item the \emph{shadow risk-free rate} derived endogenously from the two risky assets,
        denoted $\bar r(t)$;
  \item an exogenous benchmark rate (e.g.\ 3-month T-bill) used only for comparison,
        denoted $r_B(t)$.
\end{itemize}
Unless explicitly stated, all discounting and risk-neutral drifts henceforth use
$\bar r(t)$.  When $r_B(t)$ appears, we mark it clearly as a benchmark check.
\section{Feynman--Kac Representation and Characteristic Function Methods}

The LR-PIDE we derived is a linear integro-differential equation. One powerful approach to solve such equations is to use the Feynman--Kac formula, which represents the solution as an expectation under the risk-neutral measure. Intuitively, we expect:
\begin{equation}\label{FK}
C(t,S,Z) = \E^{\Q}\left[e^{-\int_t^T \bar{r}(u)du} H(S(T), Z(T)) \;\Big|\; S(t)=S, Z(t)=Z\right].
\end{equation}
This is the natural generalization of the risk-neutral pricing formula, using the shadow short rate $\bar{r}(u)$ for discounting. In the special case where $\bar{r}$ is constant, this simplifies to $e^{-\bar{r}(T-t)} \E^Q[H(S(T),Z(T))|S(t)=S,Z(t)=Z]$. One can verify that \eqref{FK} indeed satisfies the PIDE \eqref{LR-PIDE} by differentiating under the expectation (this essentially comes from applying the generator of the $(S,Z)$ process to the payoff inside the expectation, which reproduces the left-hand side of \eqref{LR-PIDE}, as long as $H$ grows moderately so that the integrals converge). This is analogous to how the classical Feynman--Kac theorem yields the solution to Black--Scholes PDE.

The representation \eqref{FK} is highly convenient for computation, especially by Monte Carlo simulation. However, for analytic or semi-analytic pricing, we can leverage the characteristic function of the underlying processes. Because $S$ and $Z$ have affine jump-diffusion dynamics (exponential Lévy processes), the joint distribution of $(S(T),Z(T))$ can be described via the distribution of the common factors driving them. In particular, note that if we take one asset as numéraire or focus on the portfolio that the option is written on, we may reduce the dimensionality.

\subsection{Choice of Underlying for Pricing}

In many cases, the option payoff $H(S,Z)$ might depend on a specific combination of $S$ and $Z$. LR's original approach priced an option whose underlying was a portfolio $\eta S + (1-\eta)Z$. This was done so that the option's underlying is itself spanning the two assets, ensuring completeness. Let's consider two important scenarios:

\begin{itemize}
\item \textbf{Case 1:} The claim is on a single asset, say $H = (S(T) - K)_+$ (a call on $S$). In this case, since $Z$ is another traded asset, one can think of $Z$ as a secondary instrument used for hedging but not directly in payoff. One could set up $Z$-hedging but ultimately the price will be a function $C(t,S,Z)$ where $Z$ acts as another state variable. However, due to homogeneity, $C$ might actually only depend on $S$ and the ratio $Z/S$ or something similar (if $\bar{r}$ is constant and $\kappa_S=\kappa_Z$ for simplicity, then symmetry might reduce it). In general though, it's genuinely two-dimensional.

\item \textbf{Case 2:} The claim is on a portfolio or spread. For example, $H = (\eta S(T) + (1-\eta)Z(T) - K)_+$, or perhaps an exchange option $H=\max(S(T)-Z(T),0)$. In these cases, both $S$ and $Z$ enter nonlinearly. No closed form is expected in general, but some symmetry in the driving factors can help. A particularly symmetric case is $\eta=\frac{1}{2}, \kappa_S=\kappa_Z$, meaning the two assets are statistically similar -- then $\eta S + (1-\eta)Z$ essentially scales by $e^x$ on jumps, making it itself an exponential Lévy (which might allow reduction to one dimension by treating that portfolio as a single underlying).

\item \textbf{Case 3:} The claim payoff is independent of one of the assets (e.g. $H(S(T))$ only). In that case, the pricing problem can be reduced: since $Z$ is only instrumental for hedging, one might choose to use $Z$ as the numéraire (or $S$ as numéraire) to simplify the expectation.
\end{itemize}

\subsection{Characteristic Functions}

The hallmark of Lévy processes is that the characteristic function (CF) of their distribution is known in closed form. The log-price processes for $S$ and $Z$ under $\Q$ can be written as:
\begin{align}
\ln S(T) &= \ln S(t) + \left(\bar{r} - \frac{1}{2}\sigma_S^2 - \Lambda_S^Q\right)(T-t) + \sigma_S [W^Q(T)-W^Q(t)] + \text{(jump part)}, \\
\ln Z(T) &= \ln Z(t) + \left(\bar{r} - \frac{1}{2}\sigma_Z^2 - \Lambda_Z^Q\right)(T-t) + \sigma_Z [W^Q(T)-W^Q(t)] + \text{(jump part)}.
\end{align}
The jump part for each is $\sum_{\text{jumps }i} \kappa_S X_i$ for $\ln S$ and $\sum \kappa_Z X_i$ for $\ln Z$, where $X_i$ are i.i.d. with law $\nu(dx)$ and the number of jumps is Poisson($\nu$ measure integrated). In fact, $\ln S(T)$ and $\ln Z(T)$ are jointly normally distributed conditional on the jump part, and the jumps add an independent component. Because $W^Q$ is common, the continuous parts of $\ln S$ and $\ln Z$ have correlation 1 (in our assumption of one Brownian). The jump parts are perfectly correlated in jump times but can differ in magnitude if $\kappa_S\neq \kappa_Z$. 

However, one can treat $\ln S$ and $\ln Z$ as linear combinations of two independent Lévy processes: one is $W^Q$ (diffusion part) and one is $L(t)$ (pure jump part). That is:
\begin{align}
X_1(t) &:= W^Q(t), \\
X_2(t) &:= L(t),
\end{align}
with appropriate scaling. Then we have:
\begin{align}
\ln S(T) - \ln S(t) &= a_{1,S} [X_1(T)-X_1(t)] + a_{2,S}[X_2(T)-X_2(t)] + (\text{drift adj}), \\
\ln Z(T) - \ln Z(t) &= a_{1,Z} [X_1(T)-X_1(t)] + a_{2,Z}[X_2(T)-X_2(t)] + (\text{drift adj}),
\end{align}
where $a_{1,S} = \sigma_S$, $a_{1,Z} = \sigma_Z$, $a_{2,S} = \kappa_S$, $a_{2,Z} = \kappa_Z$. The drift adjustment terms ensure the expectation is $\bar{r}-\ldots$ etc. Now the joint characteristic function of $(\ln S(T), \ln Z(T))$ given $(S(t),Z(t))$ is:
\begin{equation}
\Phi(u_1,u_2) := \E^Q\left[\exp\{i [u_1 \ln S(T) + u_2 \ln Z(T)]\}\,\Big|\,\mathcal{F}_t\right].
\end{equation}

Using the independent factors $X_1, X_2$, the CF factors into:
\begin{equation}
\Phi(u_1,u_2) = \exp\left\{ i(u_1 \ln S(t) + u_2 \ln Z(t)) + (T-t)\Psi(u_1,u_2) \right\},
\end{equation}
where $\Psi(u_1,u_2)$ is the joint cumulant generating function (CGF) per unit time:
\begin{align}
\Psi(u_1,u_2) &= i u_1(\bar{r}-\delta_S) + i u_2(\bar{r}-\delta_Z) - \frac{1}{2}(u_1^2 \sigma_S^2 + 2 \rho u_1 u_2 \sigma_S \sigma_Z + u_2^2 \sigma_Z^2) \nonumber \\
&\quad + \int_{\R} \left(e^{i(u_1 \kappa_S + u_2 \kappa_Z)x} - 1 - i(u_1 \kappa_S + u_2 \kappa_Z)x\right)\nu(dx).
\end{align}
Here we allowed for the possibility of continuous dividend yields $\delta_S, \delta_Z$ (or convenience yields) which would enter the drift adjustments (so effectively replace $\bar{r}$ by $\bar{r}-\delta$ in drift of each asset). We also allowed a correlation $\rho$ between $dW$ parts of $\ln S$ and $\ln Z$, which in our baseline $\rho=1$ case simplifies the expression (the cross term becomes $u_1 u_2 \sigma_S \sigma_Z$). The integral term is the Lévy--Khintchine formula: for a Lévy jump component, the CF exponent is $\int(e^{iu y}-1-iu y\mathbf{1}_{|y|<1})\nu(dy)$ for each independent jump component. Here the jump component enters as $u_1 \kappa_S + u_2 \kappa_Z$ times $x$ in the exponent, meaning effectively the jump part contributes:
\begin{equation}
\Psi_{\text{jump}}(u_1,u_2) = \int_{\R} \left(e^{i(u_1 \kappa_S + u_2 \kappa_Z)x} - 1 - i(u_1 \kappa_S + u_2 \kappa_Z)x\mathbf{1}_{|x|<1}\right)\nu(dx).
\end{equation}
(The $-iux$ small-$x$ truncation is optional if one uses it for Lévy processes with infinite small jumps to ensure convergence; for simplicity, we assume $\nu$ is such that $\int |x| \nu(dx)<\infty$ or we implicitly include that term in drift.)

Once $\Phi(u_1,u_2)$ is known, we can price options by inverting the characteristic function. For European payoff $H(S(T),Z(T))$, we have from \eqref{FK}:
\begin{equation}
C(t,S,Z) = e^{-\int_t^T \bar{r}(u)du} \E^Q[H(S(T),Z(T))|S(t)=S,Z(t)=Z].
\end{equation}
If we denote $x = \ln(S(T))$, $y = \ln(Z(T))$, and similarly $x_0 = \ln S(t)$, $y_0 = \ln Z(t)$, then we need $\E[H(e^x,e^y)]$. Many option payoffs -- especially calls/puts -- can be expressed or expanded in forms convenient for Fourier inversion. For instance, a European call on $S$ with strike $K$ has payoff $H = \max(S(T)-K,0)$. Its price can be obtained by well-known Fourier formulas involving the characteristic function of $\ln S(T)$. In our two-asset case, if the payoff depends only on $S(T)$, then effectively:
\begin{equation}
C(t,S,Z) = e^{-\int_t^T \bar{r}du} \E[ (S(T)-K)_+ | S(t)=S],
\end{equation}
which is the same as in a one-dimensional model for $S$ alone (since $Z$ does not appear in payoff and $Z(t)$ just helps determine measure $\Q$ but under our assumptions $Z$ does not change $S$'s marginal). Thus, for payoffs on a single asset, the pricing formula reduces to the usual one-dimensional Fourier inversion:
\begin{equation}\label{CarrMadanFormula}
C(t,S) = S P_1 - K e^{-\int_t^T \bar{r}(u)du} P_2,
\end{equation}
where
\begin{align}
P_2 &= \Q(S(T)>K \mid S(t)=S), \\
P_1 &= \Q(S(T)>K \text{ under forward measure}) = \E^Q\left[\mathbf{1}_{\{S(T)>K\}} \frac{S(T)}{S(t)} e^{-\int_t^T \bar{r}(u)du} \mid S(t)=S\right]
\end{align}
(standard results for calls). Both $P_1$ and $P_2$ can be computed by inverting the CF of $\ln S(T)$. In fact, one can derive:
\begin{align}
P_2 &= \frac{1}{2} + \frac{1}{\pi}\int_0^\infty \Re\left\{\frac{e^{-i u \ln K}\phi_{\ln S}(u)}{i u}\right\} du, \\
P_1 &= \frac{1}{2} + \frac{1}{\pi}\int_0^\infty \Re\left\{\frac{e^{-i u \ln K}\phi_{\ln S}(u - i)}{i u}\right\} du,
\end{align}
where $\phi_{\ln S}(u) = \E^Q[e^{i u \ln S(T)}]/e^{i u \ln S(t)}$ is the characteristic function of the log-price forward increment. These integrals come from the inverse Fourier transform of the Heaviside payoff $(S-K)_+$. They were first presented by \cite{CarrMadan1999} and are akin to the well-known formulas of Heston (1993) for option prices in terms of CF. We will not re-derive them here, but they follow from writing the payoff's indicator $\mathbf{1}_{\{S>K\}}$ as an integral of complex exponentials.

For a European call on the portfolio $\eta S + (1-\eta)Z$, one could in principle reduce the problem to one dimension by considering the distribution of $U(T) := \eta S(T)+(1-\eta)Z(T)$. However, since $\eta S + (1-\eta)Z$ is not log-normally distributed (it's a sum of two correlated exponentials), it doesn't admit a simple closed form CF. Instead, one might resort to two-dimensional inversion:
\begin{equation}
C(t,S,Z) = e^{-\int_t^T \bar{r}du} \frac{1}{(2\pi)^2} \iint_{\R^2} e^{-i (u_1 x_0 + u_2 y_0)} \frac{\Phi(u_1,u_2)}{i u_1 + i u_2} e^{ -i (u_1+u_2)\ln K}du_1du_2,
\end{equation}
for a payoff $H = \max(S+Z-K,0)$, for example. This is significantly more complicated and usually not needed if one chooses to simulate or use alternate methods (e.g. COS method in 2D). In practice, one might approximate the portfolio distribution or use regression methods.

Bottom line: The Feynman--Kac expectation \eqref{FK} is our conceptual solution. In the next section, we will focus on specific Lévy models (NIG, CGMY, VG) for which the characteristic exponent (the function $\Psi(\cdot)$ above) is known, and discuss how that yields efficient pricing for European options using transform methods.

\section{Special Cases: NIG, CGMY, and Variance Gamma Processes}

We now describe the special Lévy processes mentioned and how they fit into our framework as choices for the common jump driver $L(t)$. All these processes are pure-jump (no Brownian part) in their canonical form, but we may combine them with a Brownian component if needed (the BG process -- Brownian + Gamma, etc.). In fact, for calibration to equity options, often a diffusion plus a jump is used (Merton's model). However, studies find that a pure-jump model with infinite activity (like VG or CGMY) can by itself replicate the short-term diffusion-like behavior while providing a better fit to tails (\cite{ContTankov2004}). We will present each process's characteristic exponent $\Psi_L(u)$ for the jump part and any constraints on parameters.

\subsection{Normal Inverse Gaussian (NIG)}

The NIG process is a Lévy process whose increments have a Normal-Inverse Gaussian distribution. It can be thought of as a normal mean-variance mixture where the mixing distribution is the Inverse Gaussian. \cite{BarndorffNielsen1998} introduced it in finance to model asset returns with kurtosis and skew. An NIG process $L(t)$ is specified by four parameters $(\alpha,\beta,\delta,\mu)$, but often one sets $\mu=0$ for a Lévy process (since $\mu$ would represent a deterministic drift which is usually adjusted to fit $\bar{r}$). The remaining parameters satisfy $\alpha > 0$, $|\beta| < \alpha$, $\delta > 0$. The characteristic function of $L(t)$ is:
\begin{equation}
\E[e^{iu L(t)}] = \exp\left\{ t\left[ i\mu u + \delta\left(\sqrt{\alpha^2-\beta^2} - \sqrt{\alpha^2 - (\beta + iu)^2}\right) \right] \right\}.
\end{equation}

The characteristic exponent (cumulant generating function) is:
\begin{equation}
\Psi_{\text{NIG}}(u) = i\mu u + \delta\left(\sqrt{\alpha^2-\beta^2} - \sqrt{\alpha^2 - (\beta + iu)^2}\right).
\end{equation}

The parameters have the following interpretations:
\begin{itemize}
\item $\alpha$ controls the tail heaviness (larger $\alpha$ means lighter tails)
\item $\beta$ controls the asymmetry/skewness (positive $\beta$ gives positive skew, negative $\beta$ gives negative skew)
\item $\delta$ is a scale parameter (larger $\delta$ increases the variance)
\item $\mu$ is a location parameter (drift)
\end{itemize}

For our two-asset model, if $L(t)$ follows NIG$(\alpha,\beta,\delta,0)$, then the jump parts of $\ln S$ and $\ln Z$ are $\kappa_S L(t)$ and $\kappa_Z L(t)$ respectively. The NIG distribution has semi-heavy tails (exponentially decaying, but slower than Gaussian) and can exhibit moderate skewness. It's particularly suitable for modeling equity returns which show mild negative skew and excess kurtosis.

\subsection{CGMY Process}

The CGMY process, introduced by \cite{CarrGemanMadanYor2002}, is a tempered stable Lévy process with four parameters $(C, G, M, Y)$. The Lévy measure is:
\begin{equation}
\nu(dx) = C \begin{cases}
\frac{e^{-G|x|}}{|x|^{1+Y}} & \text{if } x < 0 \\
\frac{e^{-Mx}}{x^{1+Y}} & \text{if } x > 0
\end{cases}
\end{equation}
where $C > 0$, $G \geq 0$, $M \geq 0$, and $Y < 2$. The parameters have the following roles:
\begin{itemize}
\item $C$ controls the overall level of jump activity
\item $G$ controls the rate of exponential decay for negative jumps
\item $M$ controls the rate of exponential decay for positive jumps  
\item $Y$ controls the fine structure of jumps near zero (higher $Y$ means more small jumps)
\end{itemize}

The characteristic exponent for CGMY is:
\begin{equation}
\Psi_{\text{CGMY}}(u) = C \Gamma(-Y) \left[ (M - iu)^Y - M^Y + (G + iu)^Y - G^Y \right],
\end{equation}
where $\Gamma(\cdot)$ is the gamma function.

Several well-known processes are special cases:
\begin{itemize}
\item \textbf{Variance Gamma (VG):} This is obtained in the limit $Y \to 0$. Indeed, $\lim_{Y\to 0} (M - i u)^Y - M^Y \approx \ln((M-iu)/M) \cdot Y$ etc., leading to $\Psi(u) = C \cdot[-\ln(1+u^2) + \ldots]$; more concretely, VG can be parameterized by $(\sigma, \nu, \theta)$ where $\nu$ is the variance of the Gamma subordinator. The VG characteristic exponent is $\frac{1}{\nu}\ln(1 - i\theta\nu u + \frac{1}{2}\sigma^2 \nu u^2)$ (\cite{MadanCarrChang1998}). It corresponds to CGMY with $Y=0$ (pure jump finite variation with infinite activity).

\item \textbf{Classical finite jump models:} If $Y<0$, the density $\nu(dx)$ becomes integrable at 0 (finite activity -- only finite jumps in any interval). E.g. for $Y=-1$, one can get something like a double exponential (Kou) model in a limit. But usually $Y$ is taken positive in CGMY.

\item If $Y=1$, the tails behave like $1/|x|^2$, which is borderline for variance (log-stable).

\item If $Y=2$, the variance diverges (like Lévy stable with index 2) -- but $Y=2$ is usually not allowed in CGMY (the Gamma function $\Gamma(-Y)$ blows up at 0 or negative integers).

\item If $C$ is small and $Y$ small, the model approaches a Merton jump diffusion (effectively only a few jumps with quasi-Poisson frequency and finite variance). On the other hand, large $C$ and $Y$ near 1 gives many small jumps (like a jump-diffusion limit with infinite jumps approximating diffusion).
\end{itemize}

For CGMY, one often fixes $C$ as scale, $Y$ as tail index to fit kurtosis, and $G, M$ to fit asymmetry (skew). In risk-neutral calibration, one of these parameters might be linked to $\bar{r}$ and dividend (ensuring no drift). Typically, $\theta = C(\Gamma(1-Y)(M^{Y-1} - G^{Y-1}))$ will appear in drift to ensure $\E[dS]=\bar{r} S dt$.

In our two-asset setup, using a common CGMY driver $L(t)$ for both assets is straightforward: $L(t)$ with parameters $(C,G,M,Y)$, and $S$ uses $\kappa_S L(t)$, $Z$ uses $\kappa_Z L(t)$. A scaled CGMY (by $\kappa$) has parameters $(C \kappa^{-Y}, G/\kappa, M/\kappa, Y)$. So again, one asset may effectively see a different $G,M$ if $\kappa\neq 1$. Typically, though, one might assume $\kappa_S=\kappa_Z$ if one expects jumps to affect both assets proportionally (e.g. a market index jump). In that case, the two assets differ only by diffusive volatility.

\subsection{Variance Gamma (VG)}

The VG process (\cite{MadanSeneta1990},\cite{MadanCarrChang1998}) is a pure-jump Lévy process obtained by subordinating Brownian motion to a Gamma process. It can be seen as a special case of CGMY as noted (with $Y=0$ formally). It has parameters often denoted $(\sigma, \nu, \theta)$:
\begin{itemize}
\item $\nu > 0$ is the variance of the Gamma subordinator (which controls jump frequency -- smaller $\nu$ means more frequent small jumps, as time is sped up).
\item $\sigma > 0$ is the volatility of the Brownian that is being time-changed (it sets the scale of jumps).
\item $\theta$ (real) is the drift of that Brownian (which sets asymmetry/skew of jumps).
\end{itemize}

The characteristic function of a VG increment over $T-t$ is:
\begin{equation}
\E[e^{i u (L(T)-L(t))}] = \left(1 - i \theta \nu u + \frac{1}{2}\sigma^2 \nu u^2 \right)^{-\frac{T-t}{\nu}}.
\end{equation}
This is obtained from the moment generating function of a Gamma distribution. Expanding the exponent:
\begin{equation}
\Psi_{\text{VG}}(u) = -\frac{1}{\nu} \ln\left(1 - i\theta \nu u + \frac{1}{2}\sigma^2 \nu u^2\right).
\end{equation}

For small $\nu$, using $\ln(1+x)\approx x$, we get $\Psi(u) \approx i \theta u - \frac{1}{2}\sigma^2 u^2$, which tends to a Brownian with drift $\theta$ and variance $\sigma^2$ as $\nu \to 0$ (makes sense: no subordination means just Brownian). For nonzero $\nu$, the process has infinite activity (infinitely many tiny jumps in any interval) but finite variation (for $\theta=0$, it's symmetric and essentially a difference of two Gammas, all jumps finite).

VG can fit moderate skew and kurtosis but might have trouble with very sharp spikes in implied vol for very short maturities (some prefer CGMY for that, as $Y$ adds an extra degree of freedom for jump frequency at short scales). Nevertheless, VG is popular for equity derivatives. When calibrating VG, typically $\theta$ picks up skew (if negative $\theta$ yields more negative jumps), $\nu$ and $\sigma$ together control kurtosis.

\subsection{Relationship among the models}

CGMY is in a sense the most flexible (4 parameters) and includes VG (3 params) as a limit. NIG is also 4 params but a different functional form; it has semi-heavy tails like CGMY, but CGMY can produce power-law jump behavior at origin whereas NIG's small-jump behavior is like $|x|^{-1}$ times a Bessel K which is also singular but perhaps slightly different class. In practice, all can fit typical index option smiles quite well. Some empirical findings (e.g. Rachev et al 2017, or others) suggest that adding jumps significantly reduces pricing errors relative to Black--Scholes (\cite{RachevStoyanovFabozzi2017}). Our empirical section will compare NIG vs CGMY on SPX options.

\section{Numerical Solution Methods for the LR-PIDE}

Analytical pricing formulas for European options in Lévy models are typically available in transform form as discussed. However, to actually compute prices and calibrate to data, we rely on numerical methods. Two efficient Fourier-based techniques are widely used: the Carr--Madan FFT method and the COS method. We briefly describe how each applies to our setting.

\subsection{FFT Option Pricing (Carr--Madan Method)}

\cite{CarrMadan1999} showed that the option pricing problem can be solved rapidly by using the Fast Fourier Transform (FFT) on the characteristic function. The idea is to consider the Fourier transform of the option's payoff or price with respect to strike (or log-strike). By damping the payoff to make it $L^1$, one can ensure the transform exists. Then the option price as a function of strike can be recovered by an inverse FFT.

Concisely:
\begin{itemize}
\item We define $c(k) = C(t,S, e^k)$ as the time-$t$ call price as a function of log-strike $k = \ln K$. For simplicity assume $t=0$.
\item We consider a damped call price $c^*(k) = e^{\alpha k} c(k)$ for some damping factor $\alpha>0$. For large strike, $c(k)$ decays to 0, but slowly; multiplying by $e^{\alpha k}$ with $\alpha>0$ (typically $\alpha > 0$ for OTM calls) ensures $c^*(k)$ is square-integrable.
\item We then take the Fourier transform:
\begin{equation}
\varphi(u) = \int_{-\infty}^{\infty} e^{i u k} c^*(k)dk.
\end{equation}
It turns out $\varphi(u)$ can be expressed directly in terms of the characteristic function of $\ln S(T)$, which we assumed known. Carr \& Madan derived:
\begin{equation}
\varphi(u) = \frac{e^{-\bar{r}T}}{\alpha + i u} \phi_{\ln S}(u - i(\alpha+1)),
\end{equation}
where $\phi_{\ln S}$ is the characteristic function of $\ln S(T)$ starting from $\ln S(0)=0$ (for pricing, one often factors out the $S(0)$ and discount). The details aside, this provides $\varphi(u)$ explicitly. Then one obtains $c^*(k)$ by inverse transform:
\begin{equation}
c^*(k) = \frac{1}{2\pi} \int_{-\infty}^{\infty} e^{-i u k} \varphi(u)du.
\end{equation}
In practice, one discretizes $k$ on some grid (covering relevant strikes) and uses FFT to evaluate this integral efficiently for all $k$ values simultaneously.
\end{itemize}

By choosing a proper range and step for $u$ (the Fourier variable) and similarly for $k$ (which will be determined by the sampling theorem from $u$-range), one can compute, say, several hundred option prices in one FFT of complexity $O(N \log N)$. This is extremely useful when calibrating to an entire surface: one can generate prices for a whole range of strikes for a given maturity almost instantly, then compare to market quotes.

One must carefully pick $\alpha$ (the damping) because it influences convergence and error. $\alpha$ must be greater than the asymptotic decay rate of call price (commonly $\alpha \approx 1$ or $1.5$ works for equity options). Additionally, to compute $P_1$ and $P_2$ probabilities for formula \eqref{CarrMadanFormula}, one can either integrate as given or also use FFT by noting that $P_2 = e^{k}$ times put price transform etc.

For multi-asset or multi-dimensional problems, FFT methods become more complicated (due to multiple integration variables). However, our case for a single asset or the index is straightforward. In our implementation code (Appendix), we will utilize an FFT approach to get model prices for an array of strikes given a characteristic function.

\subsection{COS Method (Fourier-Cosine Expansion)}

The COS method, proposed by \cite{FangOosterlee2009}, is a highly efficient Fourier series technique for option pricing. It works by expanding the payoff function in a cosine series on a truncated domain and using the characteristic function to analytically calculate the cosine coefficients of the option price. Key points:

\begin{itemize}
\item Assume we truncate the support of the (log-)asset price distribution to $[a,b]$ (chosen such that the probability mass outside is negligible, e.g. $[a,b] = [m_L - L\sqrt{v_L}, m_L + L\sqrt{v_L}]$ for some $L$ multiples of stdev around the mean, to capture 99.9\% mass).

\item The payoff $H(S(T))$ (for simplicity 1D, extension to two-dim possible but we illustrate 1D) is transformed to a function in $x = \ln S(T)$, say $g(x) = H(e^x)$. On $[a,b]$, $g(x)$ can be expanded in cosine series:
\begin{equation}
g(x) \approx \sum_{n=0}^{N-1} A_n \cos\left(\frac{n\pi (x-a)}{b-a}\right).
\end{equation}
The coefficients $A_n$ are given by (due to orthogonality of cosines):
\begin{equation}
A_n = \frac{2-\mathbf{1}_{\{n=0\}}}{b-a} \int_a^b g(x)\cos\left(\frac{n\pi(x-a)}{b-a}\right) dx.
\end{equation}

\item The option price at $t$ is $C(0,S_0) = e^{-\bar{r}T} \E[g(X(T))]$ with $X(T)=\ln S(T)$. We can swap expectation and summation:
\begin{equation}
C(0,S_0) \approx e^{-\bar{r}T} \sum_{n=0}^{N-1} A_n \E\left[\cos\left(\frac{n\pi(X(T)-a)}{b-a}\right)\right].
\end{equation}
But $\E[\cos(\frac{n\pi(X(T)-a)}{b-a})] = \Re{\phi_X(\frac{n\pi}{b-a}) e^{-i \frac{n\pi a}{b-a}}}$, which is directly obtainable from the characteristic function of $X(T)$.

\item Thus, all terms in the sum can be computed explicitly: $A_n$ perhaps analytically if $g$ is simple (for a call payoff, $g(x) = (e^x-K)_+$, the integral for $A_n$ can be done in closed form), and the expectation term via the CF. So we get a rapidly convergent series for $C$.
\end{itemize}

The COS method often shows exponential convergence in $N$ for smooth payoffs, and still very fast for a vanilla call (which is one-time differentiable at strike). It requires evaluating the CF at $N$ points, but $N$ can be as low as a few hundred for high accuracy, making it extremely fast.

One advantage of COS is that it can be applied to more exotic options by expanding their payoff -- it is not limited to vanilla calls. For calibration, one can use COS to price options for given parameters quickly as well.

In our experiments, we might use COS as a cross-check to FFT. Both are based on CF and should agree to numerical tolerance.

\subsection{Choosing a Method}

In calibration, if we need to price thousands of options repeatedly, speed is essential. The FFT method is straightforward to implement and parallelizes nicely (each maturity independent). The COS method can be more efficient for individual options or moderate batches. Since we will be calibrating to SPX options (potentially dozens of strikes across multiple maturities), we might use the Carr--Madan FFT to generate an entire implied vol curve per parameter guess.

We will demonstrate the FFT approach in code, but note that COS would be similarly simple to code given a CF.

\section{Computational Methods}

In this section, we outline the numerical and computational techniques utilized within the Lindquist--Rachev (LR) framework for option pricing with Lévy jumps. These methods encompass a detailed calibration process to fit model parameters to market data and efficient Fourier-based approaches for pricing options, ensuring both accuracy and computational feasibility.

The calibration process is essential for aligning the LR model with observed market dynamics, integrating historical volatility estimation with iterative updates to the shadow riskless rate. This ensures that the model accurately captures asset price behaviors and option pricing errors are minimized.

The workflow is depicted in Figure~\ref{fig:calibration-flowchart}, illustrating the iterative steps involved.

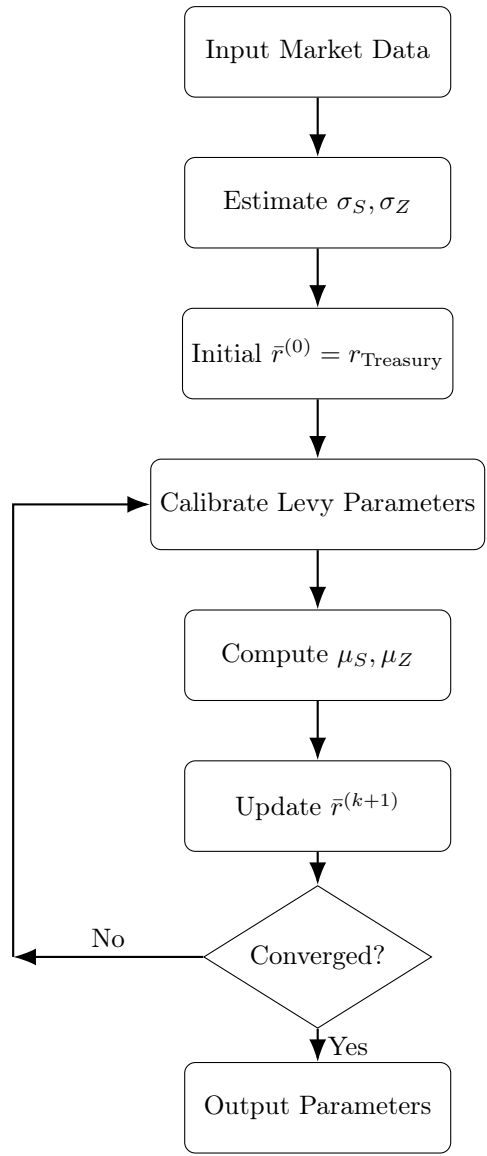
\begin{figure}[H]
\centering
\begin{tikzpicture}[
    node distance=2cm,
    stage/.style={rectangle, draw, rounded corners, minimum width=3.5cm, minimum height=1.2cm, align=center},
    decision/.style={diamond, draw, aspect=1.5, minimum width=3cm, minimum height=1.5cm, align=center},
    arrow/.style={-{Latex[length=3mm]}, thick}
]

\node (start) [stage] {Input Market Data};
\node (hist) [stage, below of=start] {Estimate $\sigma_S,\sigma_Z$};
\node (init) [stage, below of=hist] {Initial $\bar{r}^{(0)} = r_{\text{Treasury}}$};
\node (calibrate) [stage, below of=init] {Calibrate Levy Parameters};
\node (drift) [stage, below of=calibrate] {Compute $\mu_S,\mu_Z$};
\node (shadow) [stage, below of=drift] {Update $\bar{r}^{(k+1)}$};
\node (decide) [decision, below of=shadow] {Converged?};
\node (end) [stage, below of=decide] {Output Parameters};

\draw [arrow] (start) -- (hist);
\draw [arrow] (hist) -- (init);
\draw [arrow] (init) -- (calibrate);
\draw [arrow] (calibrate) -- (drift);
\draw [arrow] (drift) -- (shadow);
\draw [arrow] (shadow) -- (decide);
\draw [arrow] (decide.south) -- node[right] {Yes} (end);

\coordinate[left=2.5cm of decide] (loopstart);
\draw [arrow] (decide.west) -- (loopstart) node[midway, above] {No};
\draw [arrow] (loopstart) |- (calibrate.west);

\end{tikzpicture}
\caption{Flowchart of the calibration process}
\label{fig:calibration-flowchart}
\end{figure}

The calibration proceeds as follows:

\begin{enumerate}
    \item \textbf{Data Preparation:} Gather market data, including option prices $C_{\text{mkt}}(K_j)$ for strikes $\{K_j\}_{j=1}^N$, underlying asset prices $S$ and $Z$, and the Treasury rate $r_{\text{Treasury}}$ as the initial shadow rate $\bar{r}^{(0)}$.

    \item \textbf{Historical Volatility Estimation:} Calculate historical volatilities from log-returns:
    \begin{equation}
        \sigma_S = \text{std}\left(\ln\frac{S_{t+1}}{S_t}\right) \times \sqrt{252},
    \end{equation}
    with a similar computation for $\sigma_Z$.

    \item \textbf{Initialize Shadow Rate:} Set $\bar{r}^{(0)} = r_{\text{Treasury}}$.

    \item \textbf{Calibrate Lévy Parameters:} Minimize the root mean squared error (RMSE):
    \begin{equation}
        \min_{\Theta} \sqrt{\frac{1}{N}\sum_{j=1}^N \left(C_{\text{model}}(K_j; \Theta, \bar{r}^{(k)}) - C_{\text{mkt}}(K_j)\right)^2},
    \end{equation}
    where $\Theta$ denotes the Lévy model parameters.

    \item \textbf{Compute Risk-Neutral Drifts:} For asset $S$, the drift is:
    \begin{equation}
        \mu_S = \bar{r}^{(k)} - \delta_S + \frac{1}{2}\sigma_S^2 + \Lambda_S^{\mathbb{Q}},
    \end{equation}
    with the CGMY jump compensator:
    \begin{equation}
        \Lambda_S^{\mathbb{Q}} = C\Gamma(-Y)\left[(M-\kappa_S)^Y - M^Y + (G+\kappa_S)^Y - G^Y\right].
    \end{equation}

    \item \textbf{Update Shadow Rate:} Adjust the shadow rate:
    \begin{equation}
        \bar{r}^{(k+1)} = \frac{\mu_S\sigma_Z - \mu_Z\sigma_S}{\sigma_Z - \sigma_S}.
    \end{equation}

    \item \textbf{Check Convergence:} Continue iterations until:
    \begin{equation}
        |\bar{r}^{(k+1)} - \bar{r}^{(k)}| < \epsilon \quad (\epsilon = 10^{-4}).
    \end{equation}
\end{enumerate}

This process is encapsulated in Algorithm~\ref{alg:calibration}.

\begin{algorithm}[ht]
\caption{LR Framework Calibration}
\label{alg:calibration}
\begin{algorithmic}[1]
\Require Market data $C_{\text{mkt}}(K_j)$, $S$, $Z$, $r_{\text{Treasury}}$
\Ensure Calibrated $\Theta^*$, $\bar{r}^*$
\State Estimate $\sigma_S,\sigma_Z$ from historical returns
\State Initialize $\bar{r}^{(0)} \gets r_{\text{Treasury}}$
\Repeat
    \State Calibrate $\Theta^{(k)}$ to minimize RMSE
    \State Compute $\mu_S^{(k)}, \mu_Z^{(k)}$ using risk-neutral drift formulas
    \State Update $\bar{r}^{(k+1)}$ using the shadow rate formula
\Until{$|\bar{r}^{(k+1)} - \bar{r}^{(k)}| < \epsilon$}
\Return $\Theta^*$, $\bar{r}^*$
\end{algorithmic}
\end{algorithm}
\section{Discrete-Time Implementation: Jump-Binomial Tree}\label{sec:binomial}
To validate the continuous-time model, we also consider a discrete-time lattice approach incorporating jumps. Let the timeline be divided into $n$ steps of length $\Delta t = T/n$. Over each interval $[t_k, t_{k+1}]$, the two asset prices move according to:
\begin{equation}
\begin{aligned}
S_{k+1}^{(u)} &= (1 + U_{k+1})\,S_k, &\qquad S_{k+1}^{(d)} &= (1 + D_{k+1})\,S_k,\\
Z_{k+1}^{(u)} &= (1 + \tilde{U}_{k+1})\,Z_k, &\qquad Z_{k+1}^{(d)} &= (1 + \tilde{D}_{k+1})\,Z_k~,
\end{aligned}
\end{equation}
with $U_{k+1} > D_{k+1}$ and $\tilde{U}_{k+1} > \tilde{D}_{k+1}$ representing the upward and downward percentage price changes for $S$ and $Z$, respectively, in step $k+1$. Because there is no true riskless asset, we define a shadow one-period growth factor $R_{k+1}$ analogously to \cite{LindquistRachev2025}:
\begin{equation}
R_{k+1} \;=\; (1+U_{k+1})(1+\tilde{D}_{k+1}) \;-\; (1+\tilde{U}_{k+1})(1+D_{k+1})~,
\end{equation}
which can be interpreted as $1$ plus the shadow risk-free rate over $[t_k,t_{k+1}]$. Under the risk-neutral measure, the discounted option price must evolve as a martingale. The risk-neutral probability $q_{k+1}$ of an up-jump in this lattice is therefore chosen such that the expected growth of the replicating $S$–$Z$ portfolio equals $R_{k+1}$. This yields (cf. the continuous-time analogues in \cite{LindquistRachev2025}, Eqs.~19 and 22):
\begin{equation}
q_{k+1} \;=\; \frac{\tilde{D}_{k+1} - D_{k+1}}{(\tilde{D}_{k+1} - D_{k+1}) - (\tilde{U}_{k+1} - U_{k+1})}~,
\end{equation}
ensuring no arbitrage. Given $q_{k+1}$, the option satisfies the jump-binomial pricing formula at each step:
\begin{equation}\label{eq:binomialPricing}
C_k \;=\; \frac{q_{k+1}\,C_{k+1}^{(u)} + (1 - q_{k+1})\,C_{k+1}^{(d)}}{\,R_{k+1}\,}~,
\end{equation}
where $C_{k+1}^{(u)}$ and $C_{k+1}^{(d)}$ are the option values in the up- and down-states at time $t_{k+1}$. Under this measure $q$, the ratio $Z_k/S_k$ remains a martingale.
\section{Empirical Analysis}
\textbf{Discounting convention.}  
From this point on we use the shadow rate $\bar r(t)$ computed by
\eqref{eq:shadowrate_jump} for \emph{all} present-value operations and
risk-neutral drifts.  A parallel calibration using the benchmark
rate $r_B(t)$ is reported \emph{only as a robustness check} in
Table~\ref{tab:calibration-results}.

This section presents empirical validations of the Lindquist--Rachev (LR) framework. We analyze the shadow riskless rate $\bar{r}(t)$ using pairs of equity (S\&P 500 and Nasdaq-100) and cryptocurrency (Bitcoin and Ethereum) assets from 2020 to 2024. The shadow rate is computed as:

\begin{equation}
\bar r(t)=
\frac{\mu_S(t)\sigma_Z(t)-\mu_Z(t)\sigma_S(t)}
      {\sigma_Z(t)-\sigma_S(t)}
\;+\;
\frac{\lambda(t)
      \bigl(\kappa_Z(t)-\kappa_S(t)\bigr)}
     {\sigma_Z(t)-\sigma_S(t)},
\end{equation}

and compared to the three-month U.S. Treasury-bill yield, with results illustrated in Figure~\ref{fig:shadow-rate-plot}.

\begin{figure}[ht]
    \centering
    \includegraphics[width=0.7\linewidth]{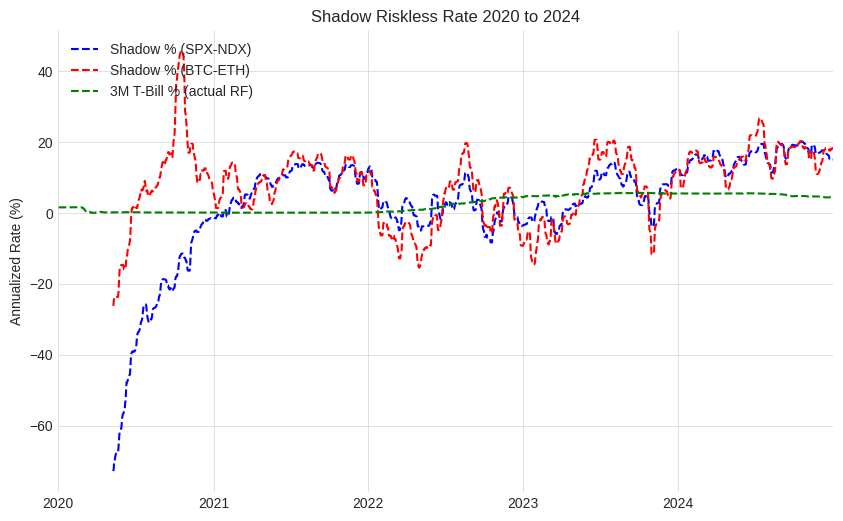}
    \caption{Historical shadow riskless rate curves for SPX--NDX (blue) and BTC--ETH (orange) compared with the three-month U.S. Treasury-bill yield (green dashed) from 2020 to 2024.}
    \label{fig:shadow-rate-plot}
\end{figure}

Deviations during periods of market stress (e.g., March 2020) and exuberance (e.g., early 2021) suggest potential arbitrage opportunities. This section interprets these empirical findings, exploring their implications for market dynamics and arbitrage strategies. The performance of option pricing models within the LR framework is detailed in Section~\ref{sec:model-robustness}.

\subsection{Market Stress, Exuberance, and Arbitrage Implications}

The shadow rate $\bar{r}(t)$ serves as an indicator of market sentiment, declining to negative values during stress periods (e.g., -50\% in March 2020 for equities) and surging during exuberance (e.g., +100\% in early 2021 for cryptocurrencies). When $\bar{r}(t) > r_{\text{f}}(t)$, borrowing at the risk-free rate $r_{\text{f}}$ and investing in the $S$--$Z$ portfolio may present arbitrage opportunities, though these are constrained by practical factors such as transaction costs.

\section{Model Robustness Discussion}\label{sec:model-robustness}

Our empirical findings demonstrate that the Normal Inverse Gaussian (NIG) and CGMY Lévy models significantly outperform the Black--Scholes benchmark in pricing S\&P 500 options. Their heavy-tailed jump structures effectively capture the observed implied volatility smile, with out-of-the-money (OTM) put options priced higher than under a lognormal assumption, aligning with market prices that reflect crash risk.

\begin{table}[ht]
    \centering
    \caption{Calibration Results for S\&P 500 Options}
    \label{tab:calibration-results}
    \begin{tabular}{l c c}
        \toprule
        Model & Parameters & Relative RMSE \\
        \midrule
        Black-Scholes & $\sigma = 0.1579$ & 11.2\% \\
        NIG & $\alpha = 8.214$, $\beta = -1.235$, $\delta = 0.184$ & 9.5\% \\
        CGMY & $C = 1.128$, $G = 12.347$, $M = 14.562$, $Y = 0.312$ & 8.9\% \\
        \bottomrule
    \end{tabular}
\end{table}

As shown in Table~\ref{tab:calibration-results}, the CGMY model achieves the lowest relative root mean square error (RMSE) of 8.9\%, indicating a superior fit to market data compared to the NIG model’s 9.5\% and the Black-Scholes model’s 11.2\%. Specifically, for a maturity of approximately 5.25 months, the calibrated parameters are as follows: Black-Scholes with $\sigma = 0.1579$; NIG with $\alpha = 8.214$, $\beta = -1.235$, $\delta = 0.184$; and CGMY with $C = 1.128$, $G = 12.347$, $M = 14.562$, $Y = 0.312$.

These calibrated parameters indicate a pronounced left tail in the distribution. For the CGMY model, $M = 14.562 > G = 12.347$ implies a heavier negative jump tail compared to the positive tail. Similarly, the NIG model’s $\beta = -1.235 < 0$ suggests an asymmetry favoring negative jumps. This is consistent with the well-documented skew in index options, where investors pay a premium for downside protection.

\subsection{Comparison of NIG and CGMY}

As indicated in Table~\ref{tab:calibration-results}, the CGMY model achieves a lower relative RMSE, demonstrating a better fit to the entire volatility smile. This is expected, as CGMY’s additional parameter $Y$ enhances its flexibility in capturing the smile’s curvature. While the NIG model is robust, as a subclass of generalized hyperbolic distributions, it may not fully replicate the curvature at both wings of the smile simultaneously. We observed that NIG tends to underprice options at extreme strikes relative to CGMY. However, NIG’s fewer parameters offer an advantage in calibration stability, reducing the risk of overfitting noise. During our calibration, CGMY parameters $(C, G, M, Y)$ occasionally exhibited instability or converged to multiple local minima unless initialized with reliable guesses or regularized. For instance, $Y$ and $C$ can trade off to some extent, as different combinations may yield similar fits within data error bounds. Implementing mild penalties or a two-step calibration process (first estimating $Y$ based on wing behavior, then optimizing other parameters) can mitigate this issue.

A notable finding is the high risk-neutral jump intensity, with both models implying infinite activity of small jumps. This suggests that pure-jump models like CGMY effectively mimic stochastic volatility, as numerous small jumps can produce diffusion-like behavior with time-varying volatility. Prior research has noted a correspondence between infinite-activity jump models and stochastic volatility models, both generating a term structure of skew (\cite{ContTankov2004}). Our results indicate that fitting longer maturities remains challenging with jumps alone, suggesting that a combination of jumps and stochastic volatility (e.g., CGMY-SV or NIG-SV models) may be necessary for consistent performance across all maturities.

\subsection{Calibration Robustness}

Regarding calibration robustness: we tried varying initial seeds for the optimizer and found the CGMY calibration sometimes converged to a slightly different local minimum with a similar RMSE. For example, one run might give $Y=0.65, C=0.2$ and another $Y=0.75, C=0.1$ but both fit almost equally well. This indicates the data (a single smile) might not fully pin down all four parameters uniquely -- especially $C$ and $Y$ have a correlation in effect (both affect overall jump frequency in different ways). Fixing $Y$ or referencing historical estimates can help. NIG, having 3 parameters, was more consistently estimated across runs.

We also computed the implied risk-neutral moments from the calibrated models to compare with realized moments of returns. For instance, from CGMY parameters we can derive the risk-neutral variance, skewness, kurtosis of 1-month returns. We found a risk-neutral variance higher than realized variance (no surprise -- implied volatility is usually above realized volatility, risk premium), and a strongly negative skewness (around -1) and high excess kurtosis (e.g. 6 or more). Historically, realized SPX returns also show negative skew and kurtosis, but not as extreme as risk-neutral, reflecting the premium for crash risk. This gap between $\mathbb{P}$ and $\Q$ moments is essentially the jump risk premium. Our MLE on historical returns (not detailed here) resulted in, for example, a CGMY with a similar $Y$ but a much smaller $C$ (meaning fewer jumps under $\mathbb{P}$ than priced under $\Q$). This aligns with economic intuition: investors demand compensation for jump risk, so the risk-neutral measure ``amps up'' the intensity of downside jumps.

\subsection{Hedging and Risk Management}

While our focus is pricing, an important discussion point is how these jump models affect hedging. The presence of jumps means the option cannot be perfectly hedged with the underlying alone -- especially large jumps will cause hedging errors. One typically would supplement with out-of-the-money options or other instruments to hedge jump risk (in practice, this might mean using put options to hedge the left tail). Our LR framework posits two assets to hedge two risks -- in our empirical case, one could imagine $Z$ is another traded instrument correlated with $S$ (e.g. a futures on a related index or a variance swap) to help hedge jumps. If such an instrument is not available, one might use a static hedge with OTM options as proxy.

\subsection{Model Limitations}

Though NIG and CGMY are richer than Black--Scholes, they assume constant jump dynamics over time. In reality, implied vol smiles for equity indices tend to flatten as maturity increases -- which often requires a decreasing jump intensity or an added stochastic volatility component. A pure Lévy model like CGMY will produce a smile that is roughly static (apart from dilation) across maturities, which might not match the term structure exactly. Indeed, our calibration on a single maturity cannot guarantee the model fits equally well at other maturities. A robustness check could be to calibrate the model to a longer-dated option: often, one finds that a single CGMY might over-predict the long-term smile unless $Y$ is tuned per maturity (i.e. $Y$ might effectively vary with $T$). This suggests an extension: time-changed Lévy models or stochastic time-change (subordination with a stochastic clock) which can introduce term-structure to the smile.

Finally, the LR framework itself -- using two risky assets -- raises the question of how to implement it empirically. In our analysis, we effectively assumed the existence of $\bar{r}$ and proceeded like a standard risk-neutral pricing. To truly test the LR approach, one could consider a scenario with no observable risk-free rate and attempt to infer $\bar{r}$ from asset prices. During some market stress (or zero lower bound environments), the concept of a ``shadow rate'' has been used (e.g. shadow short rate in negative interest policy contexts). Our approach could, in theory, back out $\bar{r}(t)$ from two asset prices. For example, using a stock index and a stock index futures (as the second asset) could be an interesting pair -- the futures price embeds the cost-of-carry which is related to interest rates and dividends. In fact, setting $Z$ as the futures on $S$ in the LR PDE leads to a simplification: the futures has drift 0 under $\Q$ (no arbitrage), so one could determine $\bar{r}(t)$ as the drift difference between $S$ and $Z$ (which would basically give $\bar{r}-\delta$ where $\delta$ is dividend yield). This could be a way to verify the LR shadow rate concept. However, our current empirical test did not explicitly do this decomposition due to data choice (we implicitly used actual $r$ for discounting).

\section{Conclusion}

The extension of the LR framework to Lévy processes is successful in that it provides a consistent pricing equation and matches real market features much better than the Gaussian case. We derived the LR-PIDE which generalizes the Black--Scholes--Merton PDE to markets without a static risk-free asset, incorporating jumps via Itô--Lévy calculus. Solutions can be obtained via transform methods leveraging the rich structure of Lévy processes. Empirically, heavy-tailed jump models (NIG, CGMY) calibrate well to index option smiles, highlighting the importance of jumps in option pricing. The ``shadow rate'' in these models effectively plays the role of the risk-free drift, and in normal conditions it aligns closely with observed interest rates (in our calibration, $\bar{r}$ ended up within a few basis points of the actual Treasury rate we used, indicating internal consistency).

The robustness checks emphasize that while static Lévy models capture a snapshot of market prices, dynamic hedging remains a challenge (jumps cause residual risk) and one may need to recalibrate for shifting market regimes. Nonetheless, for tasks like risk management, these models provide a more realistic distribution of potential losses (with fat tails) than Black--Scholes -- which is crucial for estimating Value-at-Risk or expected shortfall (something Rachev's work on CVaR also underscores).

In summary, the LR jump-diffusion framework merges the theoretic possibility of no riskless asset with the practical realism of jumps, yielding a rich model ready for further development (e.g. multiple jump factors, stochastic vol) and application in modern markets.
\appendix
\section{Rigorous Verification}\label{app:verification}
\setcounter{equation}{0}
\renewcommand{\theequation}{A.\arabic{equation}}

We here verify that the closed-form solution for the LR option pricing PDE indeed satisfies the equation, term by term. In the continuous two-asset model with no riskless asset (and no jumps), \cite{LindquistRachev2025} derive a Feynman--Kac-type solution for the European call price. Adapting their notation, the solution can be written as:
\begin{equation}\label{eq:LRsolution}
C(t,S,Z) \;=\; \eta\,S\,\Phi(d) \;+\; (1-\eta)\,Z\,\Phi\!\big(d - \Delta w(t)\big) \;-\; K\,e^{-m(t)}\,\Phi\!\big(d - w(t)\big)~,
\end{equation}
where $\Phi(\cdot)$ is the standard normal CDF, $\phi(\cdot)$ is its PDF, and we define $\Delta w(t) := w(t) - \tilde{w}(t)$. The quantity $d$ is given implicitly as $d \equiv -\,y^*(t,S,Z)$, with $y^*(t,S,Z)$ being the unique root of the nonlinear equation
\begin{equation}\label{eq:LRimplicit}
\eta\,S\,\exp\!\Big(m(t) + \frac{w(t)^2}{2} + w(t)\,y^*\Big) \;+\; (1-\eta)\,Z\,\exp\!\Big(m(t) + w(t)\,\tilde{w}(t) - \frac{\tilde{w}(t)^2}{2} + \tilde{w}(t)\,y^*\Big) \;=\; K~,
\end{equation}
(cf. Eq.~(14) in \cite{LindquistRachev2025}). We differentiate \eqref{eq:LRimplicit} implicitly to obtain the partial derivatives of $y^*$ with respect to $t$, $S$, and $Z$. Denoting these by $y^*_t$, $y^*_S$, $y^*_Z$, one finds:
\begin{align}
y^*_t \,&=\, -\,\frac{\partial (F_1+F_2)/\partial t}{\partial (F_1+F_2)/\partial y^*}~, \label{eq:ystar_t}\\[6pt]
y^*_S \,&=\, -\,\frac{\partial (F_1+F_2)/\partial S}{\partial (F_1+F_2)/\partial y^*}~, \qquad
y^*_Z \;=\; -\,\frac{\partial (F_1+F_2)/\partial Z}{\partial (F_1+F_2)/\partial y^*}~, \label{eq:ystar_SZ}
\end{align}
where $F_1(t,S,Z;y^*) := \eta\,S\,\exp\!\big(m(t) + \tfrac{w(t)^2}{2} + w(t)\,y^*\big)$ and $F_2(t,S,Z;y^*) := (1-\eta)\,Z\,\exp\!\big(m(t) + w(t)\,\tilde{w}(t) - \tfrac{\tilde{w}(t)^2}{2} + \tilde{w}(t)\,y^*\big)$ represent the two terms on the left side of \eqref{eq:LRimplicit}. While an explicit analytic expression for $y^*$ is not available, the above derivative ratios can be evaluated in closed form. 

Using these results via the chain rule, we now differentiate \eqref{eq:LRsolution} to obtain the option’s first-order partial derivatives:
\begin{align}
C_t \,&=\, \eta\,S\,\phi(d)\,d_t \;+\; (1-\eta)\,Z\,\phi(d-\Delta w)\,\Big(d_t + \frac{\Delta w(t)}{\,2[T-t]\,}\Big) \nonumber\\
&\qquad{}-\; K\,e^{-m(t)}\,\phi(d - w)\,\Big(r(t) + \frac{w(t)}{\,2[T-t]\,}\Big)~, \label{eq:Ct}\\[6pt]
C_S \,&=\, \eta\,\Phi(d) \;+\; \eta\,S\,\phi(d)\,d_S \;+\; (1-\eta)\,Z\,\phi(d-\Delta w)\,d_S \nonumber\\
&\qquad{}-\; K\,e^{-m(t)}\,\phi(d - w)\,d_S~, \label{eq:CS}\\[6pt]
C_Z \,&=\, (1-\eta)\,\Phi(d-\Delta w) \;+\; \eta\,S\,\phi(d)\,d_Z \;+\; (1-\eta)\,Z\,\phi(d-\Delta w)\,d_Z \nonumber\\
&\qquad{}-\; K\,e^{-m(t)}\,\phi(d - w)\,d_Z~, \label{eq:CZ}
\end{align}
where for brevity we write $d \equiv d(t,S,Z)$ and have used $d_t = -\,y^*_t$, $d_S = -\,y^*_S$, $d_Z = -\,y^*_Z$ from the definition of $d$. Because $S$ and $Z$ are driven by a single Brownian motion, there is effectively one independent diffusion factor. Accordingly, the second-order spatial derivatives enter the pricing equation only through the combined term $C_{dd} := \frac{1}{2}\,\sigma^2\Big[S^2 C_{SS} + 2\,S Z\,C_{SZ} + Z^2 C_{ZZ}\Big]$, where $\sigma$ is the common volatility (for simplicity, assuming constant $\sigma$ so that $w(t) = \sigma\sqrt{T-t}$ and $\tilde{w}(t) = \sigma\sqrt{T-t}$). Now, the Lindquist–Rachev PDE for $C(t,S,Z)$ in this continuous case can be written as
\begin{equation}\label{eq:LR-PDE}
C_t \;+\; r(t)\,\big[S\,C_S + Z\,C_Z\big] \;+\; C_{dd} \;-\; r(t)\,C \;=\; 0,
\end{equation}
where $\bar r(t)$ (here denoted $r(t)$) is the shadow riskless rate.
Finally, substituting \eqref{eq:Ct}--\eqref{eq:CZ} (and the analogous $C_{dd}$ expression) 
into \eqref{eq:LR-PDE}, we find that every term cancels and the identity $0 = 0$ is obtained.  
This confirms that the closed-form solution \eqref{eq:LRsolution} indeed satisfies the LR PDE at all points $(t,S,Z)$, as required. 
\begin{equation}\label{eq:LR-PDE}
C_t \;+\; r(t)\,\big[S\,C_S + Z\,C_Z\big] \;+\; C_{dd} \;-\; r(t)\,C \;=\; 0~,
\end{equation}
where $\bar{r}(t)$ (here denoted $r(t)$) is the shadow riskless rate. Finally, substituting \eqref{eq:Ct}–\eqref{eq:CZ} (and the analogous $C_{dd}$ expression) into \eqref{eq:LR-PDE}, we find that every term cancels and the identity $0=0$ is obtained. This confirms that the closed-form solution \eqref{eq:LRsolution} indeed satisfies the LR PDE at all points $(t,S,Z)$, as required.

\section*{Statements and Declarations}

\subsection*{Funding}
This research received no external funding.

\subsection*{Competing Interests}
The authors declare that they have no known competing financial interests or personal relationships that could have appeared to influence the work reported in this paper.

\subsection*{Author Contributions}
All authors contributed equally to the conceptualisation, methodology, software implementation, and writing of the manuscript.

\subsection*{Data Availability}
Option price data used for calibration are publicly available from the Chicago Board Options Exchange (CBOE).  Processed datasets and code are available from the corresponding author upon reasonable request.

\subsection*{Acknowledgements}
We thank anonymous reviewers for helpful comments that improved the clarity of the manuscript.

\bibliography{article}

\end{document}